\documentclass[11pt]{amsart}
\usepackage{latexsym,amssymb,amsmath,youngtab}
\usepackage{amsmath,amsthm,amssymb,amscd}
\usepackage[mathscr]{eucal}
\usepackage{verbatim}
\usepackage{color}
\newtheoremstyle{custom}% name
  {3pt}%      Space above
  {3pt}%      Space below
  {\slshape}%         Body font
  {}%         Indent amount (empty = no indent, \parindent = para indent)
  {\bfseries}% Thm head font
  {.}%        Punctuation after thm head
  { }%     Space after thm head: " " = normal interword space;
   {}%         Thm head spec (can be left empty, meaning `normal')
\theoremstyle{custom}
\newtheorem{theorem}{Theorem}[section]

\newtheorem{proposition}[theorem]{Proposition}
\newtheorem{proposition/definition}[theorem]{Proposition/Definition}
\newtheorem{lemma}[theorem]{Lemma}
\newtheorem{corollary}[theorem]{Corollary}

\theoremstyle{definition}

\newtheorem{example}[theorem]{Example}

\theoremstyle{remark}
\newtheorem{remark}[theorem]{Remark}

\newcommand{\stack}[2]{\ensuremath{\genfrac{}{}{0pt}{}{#1}{#2}}} 
\newtheoremstyle{exercise}% name
  {3pt}%      Space above
  {6pt}%      Space below
  {}%         Body font
  {}%         Indent amount (empty = no indent, \parindent = para indent)
  {\bfseries}% Thm head font
  {:}%        Punctuation after thm head
  { }%     Space after thm head: " " = normal interword space;
   {}%         Thm head spec (can be left empty, meaning `normal')
\theoremstyle{exercise}
\newtheorem{exercise}[theorem]{Exercise}
\newtheoremstyle{exercises}% name
  {3pt}%      Space above
  {6pt}%      Space below
  {}%         Body font
  {}%         Indent amount (empty = no indent, \parindent = para indent)
  {\bfseries}% Thm head font
  {:}%        Punctuation after thm head
  {\newline}%     Space after thm head: " " = normal interword space;
   {}%         Thm head spec (can be left empty, meaning `normal')

\theoremstyle{exercise}
\newtheorem{exercises}[theorem]{Exercises}

\input epsf
\def\boxit#1{\vbox{\hrule height1pt\hbox{\vrule width1pt\kern3pt
  \vbox{\kern3pt#1\kern3pt}\kern3pt\vrule width1pt}\hrule height1pt}}

\def\trank{\text{rank}}

\def\bv{\bold v}

\def\BC{\mathbb C}

\def\BP{\mathbb P}

\def\tdim{{\rm dim}}

\def\hd{,...,}
\def\ww{\wedge}
\def\upperp{{}^\perp}

\def\cO{{\mathcal O}}

\def\11{\mathbf 1}

\def\a{\alpha}

\def\o{\omega}

\def\b{\beta}

\def\s{\sigma}

\def\ot{{\mathord{ \otimes } }}
\def\op{{\mathord{\,\oplus }\,}}

\def\ra{{\mathord{\;\rightarrow\;}}}

\def\dim{{\rm dim}\;}
\def\La#1{\Lambda^{#1}}

\def\op{\oplus}

\def\ep{\epsilon}
\def\op{\oplus}

\def\s{\sigma}

\def\a{\alpha}
\def\b{\beta}

\def\FS{\mathfrak  S}

\def\BP{\mathbb  P}
\def\BC{\mathbb  C}

\def\ep{\epsilon}

\def\hd{, \hdots ,}

\def\La#1{\Lambda^{#1}}

\def\ur{\underline {\bold R}}

\def\ra{\rightarrow}

\def\tim{\operatorname{Image}}
\def\tdim{\operatorname{dim}}
\def\tker{\operatorname{ker}}
\def\tlim{\lim}

\def\trank{\operatorname{rank}}

\def\upperp{{}^{\perp}}

\def\ww{\wedge}

\def\bbb{{\bold{b}}}

\def\be{\begin{equation}}
\def\ene{\end{equation}}
\def\aaa{{\bold {a}}}
\def\bbb{{\bold {b}}}
\def\ccc{{\bold {c}}}

\def\tlog{{\rm{log}}}

\def\Mn{M_{\langle \nnn,\nnn,\nnn\rangle}}
\def\Mnl{M_{\langle \mmm,\nnn,\lll\rangle}}
\def\Mnnl{M_{\langle \nnn,\nnn,\lll\rangle}}

\makeatletter
\def\Ddots{\mathinner{\mkern1mu\raise\p@
\vbox{\kern7\p@\hbox{.}}\mkern2mu
\raise4\p@\hbox{.}\mkern2mu\raise7\p@\hbox{.}\mkern1mu}}
\makeatother

\def\trank{{\mathrm {rank}}}

\def\aaa{{\bold a}}\def\bbb{{\bold b}}\def\ccc{{\bold c}}

\def\mmm{\bold m}\def\nnn{\bold n}\def\lll{\bold l}
 
\def\rig#1{\smash{ \mathop{\longrightarrow}
    \limits^{#1}}}

\begin{document}

\title{New lower bounds for the border rank of matrix multiplication}
\author{J.M. Landsberg and Giorgio Ottaviani}
%\date{}
 \begin{abstract} 
 The border rank of  the matrix multiplication operator for $\nnn\times \nnn$ matrices is a standard measure
of its complexity. Using techniques from algebraic geometry and representation theory,  we show  the border rank 
 is at least
{$2\nnn^2-\nnn$.} Our bounds are better than 
 the previous lower bound (due to Lickteig in 1985) of  $\frac 32\nnn^2+\frac{\nnn}{2}-1$  for
 all  $\nnn\geq 3$. {The bounds are obtained by finding new equations that bilinear maps of small border
rank must satisfy, i.e.,  new equations for secant varieties of triple Segre products, that matrix multiplication
fails to satisfy.}
\end{abstract}
\thanks{ Landsberg  supported by NSF grant  DMS-1006353, Ottaviani is member of GNSAGA-INDAM}
\email{jml@math.tamu.edu, ottavian@math.unifi.it}
\maketitle

\section{Introduction  and {statement} of results}

Finding lower bounds in complexity theory is considered difficult. For example, chapter 14 of \cite{MR2500087} is \lq\lq Circuit lower bounds: Complexity theory's Waterloo\rq\rq .
The complexity of matrix multiplication is roughly equivalent to the complexity of many standard operations in linear algebra,
such as taking {the} determinant or inverse of a matrix. A standard measure of the complexity of an operation is the minimal 
%length of a  straight line program (or 
{size of an } arithmetic circuit needed to perform {it.} 
%Another measure is just to count the number of multiplications performed.
 The exponent of matrix multiplication $\o$ is defined to be   $\underline{\lim}_{\nnn}\tlog_{\nnn}$  of
 the {arithmetic cost to multiply}  $\nnn\times \nnn$  matrices, or equivalently,
 $\underline{\lim}_{\nnn}\tlog_{\nnn}$  of the minimal number of multiplications needed \cite[Props. 15.1,  15.10]{bucs:96}.  
 %(The result that these are equivalent justifies ignoring additions.)
{Det}ermining the complexity of matrix multiplication is a central
 question of practical importance.  We give new lower bounds
 for its complexity in terms of border rank. {These lower bounds are used to
 prove further lower bounds for tensor rank in \cite{Lan1206.1530,Masrav}.}
 
% The rank one bilinear maps are those that can be executed
% using just one scalar multiplication. 
% The rank of a bilinear map $T$   {is} the smallest $r$ such that $T$ can
% be written as a sum of $r$ rank one bilinear maps. 

{L}et $A,B,C$ be vector spaces,
with dual spaces $A^*,B^*,C^*$,
and let $T: A^*\times B^*\ra C$ be a bilinear map. The {\it rank} of $T$  is the smallest $r$ such that there
exist $a_1\hd a_r\in A$, $b_1\hd b_r\in B$, $c_1\hd c_r\in C$ such that
$T(\a,\b)=\sum_{i=1}^r a_i(\a)b_i(\b)c_i$.  The {\it border rank}  of $T$ is the smallest $r$ such that $T$
   can be written as a limit of a sequence of bilinear maps
 of   rank $r$. Let 
 $\ur(T)$ denote the border rank of $T$. {See \cite{Ltensor} or \cite{bucs:96} for more on the rank and border rank of tensors, especially
 the latter for their relation to other measures of complexity.}
 
 Let  $\Mnl:  Mat_{\mmm\times \nnn}
 \times Mat_{\nnn\times \lll}\ra Mat_{\mmm\times\lll}$ 
 denote the matrix multiplication operator. One has (see, e.g., \cite[Props. 15.1,  15.10]{bucs:96}) that 
$\o=\underline{\tlim}_{\nnn}(\tlog_{\nnn}\ur(\Mn))$. 
%  For more on the relation 
 %between border rank and other measures of complexity,
% see \cite{bucs:96}.  
 {N}a\"\i vely   $\ur(\Mnl)\leq \mmm\nnn\lll$ via the standard
 algorithm. In 1969,  V. Strassen \cite{Strassen493} showed that ${\bold R}(M_{\langle 2,2,2\rangle})\leq 7$
 and, as a consequence,  {  ${\bold R}(\Mn)\leq \cO(\nnn^{2.81})$.}  Further upper bounds have been derived since then
by numerous authors, with
the current record {${\bold R}(\Mn)\leq \cO(\nnn^{2.3727})$ \cite{williams}}. In
 1983 {Strassen} showed \cite{Strassen505} that $\ur(\Mn)\geq \frac 32\nnn^2${,} and shortly thereafter
 T. Lickteig {\cite{MR86c:68040}}
%\cite{MR87f:15017}
 showed  $\ur(\Mn)\geq \frac 32\nnn^2+ \frac{\nnn}{2}-1$.  Since then
 no further general lower bound had been found (although it
 is now known that $\ur(M_{\langle 2,2,2\rangle})= 7$, see \cite{Lmatrix, HIL}),
and a completely different proof (using methods proposed by
Mulmuley and Sohoni for Geometric complexity theory) that
$\ur(M_{\langle \nnn,\nnn,\nnn\rangle})) \geq \frac 32 \nnn^2 - 2$
was given in \cite{peterchristian}. 

Our results are as follows:

\begin{theorem}\label{newmainthm} 
Let $\nnn\le\mmm$. For all $\lll\ge 1$
 \be\label{rboundnew} 
\ur(\Mnl)\geq  \frac{{\nnn\lll}\left(\nnn+\mmm-1\right)}{\mmm}.
\ene
\end{theorem}

\begin{corollary}\label{mainthm2} 
\be\label{rboundx}
\ur(\Mnnl)\geq { 2\nnn\lll - \lll}
\ene
 \be\label{rbound2}
\ur(\Mn)\geq  {2\nnn^2-\nnn}.
\ene 
\end{corollary}

{For} $3\times 3$ matrices, the state of the art is {now} $15\leq \ur(M_{\langle 3,3,3\rangle})\leq 21$, the
upper bound is due to Sch\"onhage \cite{MR623057}.

 \begin{remark} The best lower bound
   for the {\it rank} of matrix multiplication, 
{was, until recently,}    
%$\bold R(M_{\langle \nnn,\mmm,\lll\rangle})\geq\lll\mmm+\mmm\nnn + \lll-\mmm+\nnn-3$,  
% $\bold R(M_{\langle \nnn,\nnn,\lll\rangle})\geq  2\lll\nnn -\lll +2\nnn-2$, and 
$\bold R(M_{\langle \nnn,\nnn,\nnn\rangle})\geq \frac 52\nnn^2-3\nnn$,  due to
 to  Bl\"aser,
 % the first two are in \cite{Bl2}, and  the third in 
 \cite{Bl1}. 
 After this paper {was posted on arXiv}, and using Theorem \ref{mainthm2}, it was shown 
that  $\bold R(M_{\langle \nnn,\nnn,\nnn\rangle})\geq 3\nnn^2-4\nnn^{2/3}-\nnn$ by Landsberg in \cite{Lan1206.1530} and
{then, pushing the same methods further, {A.} Massarenti and {E.} Raviolo showed} $\bold R(M_{\langle \nnn,\nnn,\nnn\rangle})\geq 3\nnn^2-2\sqrt{2}\nnn^{3/2}-3\nnn$ in \cite{Masrav}.
 \end{remark}

Our bounds come  from explicit equations that bilinear maps of low
border rank must satisfy.  These equations are best expressed in
the language of tensors.
Our method is similar in nature to the method used by Strassen
to get his lower bounds - we find explicit polynomials that tensors
of low border rank must satisfy, and show that matrix multiplication fails
to satisfy them. Strassen found his equations via   linear algebra - taking the commutator
of certain matrices.
We found ours using representation theory and algebraic geometry. (Algebraic geometry is not needed for presenting the results.
For its role in our method see   \cite{LOannali}.)
 More precisely,   in {\S\ref{neweqnssect} } 
 we define, for every $p$, a linear map 
\be\label{yfmap}
 (M_{\langle \mmm,\nnn,\lll\rangle })_{A}^{\ww p}\colon 
\BC^{\nnn\lll{{\mmm\nnn}\choose p}}\to \BC^{\mmm\lll{{\mmm\nnn}\choose {p+1}}} 
\ene
 and we prove that
$\ur(\Mnl)\geq{{{\mmm\nnn}-1}\choose p}^{-1}\textrm{{rank}} \left[(M_{\langle \mmm,\nnn,\lll\rangle })_{A}^{\ww p}\right]$. 
In order to prove Theorem \ref{newmainthm}, we specialize this map for a judiciously chosen $p$ to a subspace where it becomes injective.
The above-mentioned  
equations are the minors of the linear map $(M_{\langle \mmm,\nnn,\lll\rangle })_{A}^{\ww p}$. 

The {map}  \eqref{yfmap} is of interest in its own right, we discuss it in detail in \S\ref{MMsect}.  
 { This} is  done with the help of   representation 
theory   - we explicitly describe the kernel as  a sum of irreducible representations
 labeled by   {Young diagrams. }
  
\smallskip

\begin{remark} 
%Viewing matrix multiplication as a map $\BC^{{3\nnn^2}}\ra \BC$,
%it is generally expected in the computer science community to have a lower bound  on the border rank  asymptotically
%like  ${3\nnn^2}$, the \lq\lq input size\rq\rq . On the other hand it is also conjectured  
It is conjectured in the computer science community that $\ur(\Mn)$ grows like $\cO(\nnn^{2+\ep})$
for any $\ep>0$.
 A truly significant lower bound  would be a function that
grew like {${ \nnn^2 h( \nnn )}$}  where $h$ is an increasing function. 
No 
%such 
{s}uper-linear 
lower bound on the complexity of any explicit tensor (or any computational problem) is known, see
\cite{MR2500087,MR2334207}.

From a mathematician's perspective, all known equations for secant varieties of Segre varieties that
have a geometric model   arise  by  translating multi-linear algebra to linear algebra, and it appears that
the limit of this technique is roughly the {\lq\lq input size\rq\rq\ $3n^2$}. 
\end{remark}

\begin{remark} The methods used here should be applicable to    
lower bound problems coming from the {\it Geometric Complexity Theory} (GCT) introduced by
Mulmuley and Sohoni \cite{MS1}, in particular to 
separate the determinant (small weakly skew circuits) from polynomials
with small formulas (small tree circuits).
\end{remark}

\subsection*{Overview}  
In \S\ref{neweqnssect} we describe the new equations to test for border {rank}   in the language of tensors.
{Theorem  \ref{newmainthm}  is proved in \S\ref{nmainpf}.
We give a detailled analysis of the kernel of the map \eqref{yfmap} in sections 
\S\ref{MMsect} and \S\ref{rbndpfsect}. This analysis should be very useful for future work.
We    conclude in \S\ref{licksect}
with a review of Lickteig's method for purposes of comparison.
An appendix \S\ref{repapp} with basic facts from representation theory that we use is included for
readers not familiar with the subject.}

\subsection*{Acknowledgments} We thank K. Mulmuley and A. Wigderson for discussions regarding the
perspective of   computer scientists,  P. B\"urgisser and A. Wigderson for 
help improving the exposition, {J. Hauenstein with help with  computer calculations,}   and M. Bl\"aser for help with the literature.

\section{The new equations}\label{neweqnssect}
Let $A,B,C$ be complex vector spaces of dimensions $\aaa,\bbb,\ccc$, with $\bbb\leq\ccc$, and
with dual vector spaces $A^*,B^*,C^*$.
Then $A\ot B\ot C$ may be thought of as the space of
bilinear maps $A^*\times B^*\ra C$. 

The most na\"\i ve equations for {border rank}  are the so-called
{\it flattenings}. Given $T\in A\ot B\ot C$, {consider $T$ as a linear map
$B^*\ra A\ot C$ and write $T_B$ for this map}. Then $\ur(T)\geq \trank (T_{B})$ and similarly for
the analogous {$T_A,T_C$}.
The rank of a linear map is determined by taking minors.

\subsection{Strassen's equations}
Strassen's equations 
\cite{Strassen505} may be understood as follows (see \S\ref{originsect} for the geometric origin of this perspective). 
As described in \cite{MR2554725},  tensor {$T_B$} with $Id_A$ to obtain {a linear map $B^*\ot A\ra A\ot A\ot C$ and 
skew-symmetrize the $A\ot A$ factor} to obtain a map
$$
T_A^{\ww 1}: B^*\ot A\ra \La 2 A\ot C.
$$
If $T$ is generic, then { one can show that}  $T_A^{\ww 1}$ will have maximal rank, and if $T=a \ot b \ot c $ is of rank one,  $\trank((a \ot b \ot c )_A^{\ww 1})= \aaa-1$. To see this, expand $a=a_1$ to a basis $a_1\hd a_{\aaa}$ of
$A$ with dual basis $\a^1\hd \a^{\aaa}$ of $A^*$. Then   $T_A^{\ww 1}=\sum_i [\a^{i}\ot b ]\
\ot [a_1\ww  a_{i}\ot c ]$,  so the image is isomorphic to $(A/a_1)\ot c$.

It follows that $\ur(T)\ge\frac{\trank(T_A^{\ww 1})}{\aaa-1}$.
Thus the best bound one could hope for with this technique is
up to $r=\frac{\bbb\aaa}{\aaa-1}$. The minors of {size}  $r(\aaa-1)+1$ of $T_A^{\ww 1}$ give
equations for {the tensors of border rank at most $r$ in $A\ot B\ot C$.} 
This is most effective when $\aaa=3$.  

When $\aaa>3$, for each $3$-plane $A'\subset A$,
consider  the restriction $T|_{A'\ot B\ot C}$ and the corresponding
equations, to obtain   {e}quations for {the tensors of border rank at most $r$ in $A\ot B\ot C$   
as long as}   $r\leq \frac{3\bbb}{2}$.  This procedure
is called {\it inheritance} (see \cite[\S 7.4.2]{Ltensor}).
 
We consider the following generalizations: {tensor  $T_B$  with $Id_{\La pA}$ to obtain a linear map $B^*\ot \La p A\ra \La p A\ot A\ot C$ and 
skew-symmetrize the $\La p A\ot A$ factor  to obtain a map}
\be\label{tadef}
T_A^{\ww p}: B^*\ot \La p A\ra \La {p+1} A\ot C.
\ene
To avoid redundancies, assume $\bbb\leq \ccc$ and $p\leq \lceil \frac \aaa 2\rceil -1$.
Then, if $T=a \ot b \ot c$ is of rank one,
$$
\trank((a \ot b \ot c )_A^{\ww p})= \binom{\aaa -1}p.
$$
To see this, compute   $T_A^{\ww p}=\sum [\a^{i_1}\ww\cdots\ww\a^{i_p}\ot b ]\
\ot [a_1\ww  a_{i_1}\ww\cdots \ww a_{i_p}\ot c ]$,  {to conclude}  the image is isomorphic to $\La p(A/a_1)\ot c$.

In summary:
\begin{theorem}\label{resume}
Interpret $T\in A\ot B\ot C$ as a linear map $B^{*}\to A\ot C$ and let 
$T_A^{\ww p}: B^*\ot \La p A\ra \La {p+1} A\ot C$ be the map obtained by skew-symmetrizing $T\ot Id_{\La p A}$ in the
$A\ot \La p A$ factor. Then
$$\ur(T)\ge\frac{\trank T_A^{\ww p} }{ \binom{\aaa -1}p}.
$$
\end{theorem}
\begin{proof}
Let $r=\ur(T)$ and let $T_{\ep}=\sum_{i=1}^rT_{\ep, i}$ be such that  $\bold R(T_{\ep,i})=1$ and $\tlim_{\ep\ra 0}T_{\ep}=T$.
Then 
$$\trank T_A^{\ww p}\le \trank (T_{\ep})_A^{\ww p}
\leq \sum_{i=1}^r\trank (T_{\ep, i})_A^{\ww p}
=r\binom{\aaa -1}p
$$
\end{proof}
 
 \begin{remark} Alternatively, one can  compute the rank using the vector bundle techniques of \cite{LOannali}.
\end{remark}

{
When this article was posted on arXiv, we only knew that 
the 
minors of size $r{\binom{\aaa-1}p}+1$ of the maps $T^{\ww p}_A$ gave nontrivial
equations for  tensors of border rank at most $r$ in $A\ot B\ot C$ for
$r\leq 2\aaa -\sqrt{\aaa}$. Then, in \cite{Lhighbranktensor}, it was shown
they actually give nontrivial equations up to the maximum
$2\bbb -1$.
}
 
 \medskip

We record the following proposition which follows  from Stirling's formula and  the discussion above.

\begin{proposition}
The equations for the variety of { tensors of border rank at most $r$ in $A\ot B\ot C$}  obtained by taking minors
of $T_A^{\ww p}$ are of degree $r\binom{\aaa-1}p +1$. In particular, when
$r$ approaches the upper bound $2\bbb$ and $p=\lceil\frac \aaa 2\rceil -1$, the equations are
asymptotically  of degree 
 %$\frac{ 2}{\sqrt{\pi}}$
 %$\frac{ 4^{\aaa -1} \bbb }{ \sqrt{\aaa-1}}$
{$\sqrt{\frac   2 \pi}$} 
 $\frac{ 2^{\aaa } \bbb }{ \sqrt{\aaa-1}}$.
\end{proposition}
  
  Theorem \ref{newmainthm}
is obtained by applying the inheritance principle  to   the case
of an $(\nnn+\mmm-1)$-plane $A'\subset A=\BC^{\nnn\mmm}$.

\subsection{Origin of the equations corresponding to minors of \eqref{tadef}}\label{originsect}
This subsection is not used in the proof of the main theorem.
We work in projective
space as the objects we are interested in are invariant under rescaling.

Let $Seg(\BP A\times \BP B\times \BP C)\subset \BP (A\ot B\ot C)$
denote the Segre variety of rank one tensors and let $\s_r(Seg(\BP A\times \BP B\times \BP C))$ denote
its $r$-th secant variety, the variety of tensors of border rank at most $r$.

In \cite{LOannali} we { introduced}  a generalization of  flattenings,
called  {\it Young flattenings}, which in the
present context is as follows: 
Irreducible polynomial representations of the 
general linear group $GL(A)$ correspond to  partitions
$\pi=(\pi_1\hd \pi_{\aaa})$, see \S\ref{glwreps}. Let $S_{\pi}A$ denote
the corresponding $GL(A)$-module. Consider representations
$S_{\pi}A,S_{\mu}B, S_{\nu}C$, and the identity maps
$Id_{S_{\pi}A} \in S_{\pi}A\ot S_{\pi}A^*$ etc...
Then we may consider
$$
T\ot Id_{S_{\pi}A} \ot Id_{S_{\mu}A} \ot Id_{S_{\nu}A} \in  A\ot B\ot C\ot S_{\pi}A\ot S_{\pi}A^*
\ot S_{\mu}B\ot S_{\mu}B^*\ot S_{\nu}C\ot S_{\nu}
C^*
$$
We may decompose $S_{\pi}A\ot A$ according to the Pieri rule (see \S\ref{pieris}) and project to one irreducible component, say $S_{\tilde \pi}A$,
where $\tilde \pi$ is obtained by adding a box to $\pi$, and similarly for $C$, while for $B$ we may decompose $S_{\mu}B^*\ot B$ 
 and project to one irreducible component,
say $S_{\hat \mu}B^*$, where $\hat \mu$ is obtained by deleting a box from $\mu$.
  The upshot is a tensor
 $$
T'\in S_{\tilde \pi}A\ot S_{\mu}  B  \ot S_{\tilde \nu }C\ot S_{\pi}A^*\ot S_{\hat \mu} B^*\ot S_{\nu} C^*
$$ 
which we may then consider as a linear map, e.g.,
 $$
T': S_{\pi}A \ot S_{\mu} B^* \ot S_{ \nu }C \ra S_{\tilde \pi}A\ot S_{\hat \mu }B^*\ot  S_{\tilde\nu } C 
$$ 
and rank conditions on $T'$ {may} give border rank conditions on $T$.

\smallskip

{ Returning to the minors of  \eqref{tadef}, the   minors of size $t+1$ of $T_A^{\ww p}$ give modules of equations } which are contained in
\be\label{latp1}
\La{t+1}(\La pA \ot B^*)\ot \La{t+1}(\La{p+1}A^*\ot C^*)
=\bigoplus_{|\mu|=t+1,\ |\nu|=t+1}
S_{\mu}(\La p A )\ot S_{\mu'}B^*\ot S_{\nu}(\La{p+1}A^*)\ot S_{\nu'}C^*.
\ene 
{Determining    which irreducible submodules
of \eqref{latp1} actually contribute nontrivial equations  appears to be difficult.
}

\section{Proof of Theorem \ref{newmainthm}}\label{nmainpf}
Let $M,N,L$ be vector spaces of dimensions $\mmm,\nnn,\lll$.
Write $A=M\ot N^*$, $B=N\ot L^*$, $C=L\ot M^*$, so
$\aaa=\mmm\nnn$, $\bbb=\nnn\lll$, $\ccc=\mmm\lll$.
The matrix multiplication operator $M_{<\mmm,\nnn,\lll>}$
is $M_{<\mmm,\nnn,\lll>}=Id_M\ot Id_N\ot Id_L\in A\ot B\ot C$.
{(See \cite[\S 2.5.2]{Ltensor} for an explanation of this identification.)}
 Let $U=N^*$.
 {
 Then
 \be\label{mmmap}
 (M_{\langle \mmm,\nnn,\lll\rangle })_{A}^{\ww p}\colon 
L\otimes U\ot\wedge^p(M\ot U)\to L\ot M^*\ot\wedge^{p+1}(M\ot U).
\ene 
This is just the identity map on the $L$ factor, so we may
write $M_{\mmm,\nnn,\lll}^{\ww p}=\psi_p\ot Id_L$, where 
\be\label{psip}\psi_p: \La p (M\ot U)\ot U\ra M^*\ot \La{p+1}(M\ot U).
\ene
}

{ The essential idea  } is to choose a subspace $A'\subset M\ot U$  on which the \lq\lq restriction\rq\rq\  of $\psi_p$ becomes injective for $p=\nnn-1$.
Take a vector space $W$ of dimension $2$, and fix
isomorphisms $U\simeq S^{\nnn-1}W^{*}$, $M\simeq  S^{\mmm-1}W^{*}$  . Let $A'$ be the {$SL(W)$-}direct summand 
$S^{\mmm+\nnn-2}W^{*}\subset S^{\nnn-1}W^{*}\ot S^{\nnn-1}W^{*}=M\ot U$.

Recall  that $S^{\alpha}W$ may be interpreted as the space of homogenous polynomials of degree $\a$ in   two variables.
If $f\in S^{\alpha}W$ and $g\in S^{\beta}W^{*}$ (with $\beta\le\alpha$) then we can perform  the contraction $g\cdot f\in S^{\alpha-\beta}W$.
In the case $f=l^\alpha$ is the power of a linear form $l$, then the contraction $g\cdot l^\alpha$
equals $l^{\alpha-\beta}$ multiplied by the value of $g$ at the point $l$,
so that  $g\cdot l^\alpha=0$ if and only if $l$ is a root of $g$.

Consider the natural skew-symmetrization map
\be\label{r1} A'\otimes\wedge^{\nnn-1}(A') \rig{} \wedge^{\nnn}(A').
\ene
{Because $SL(W)$ is reductive, there is a unique   $SL(W)$-}complement $A''$ to $A'$, so the projection $M\ot U\ra A'$ is well defined. Compose \eqref{r1} with the projection
\be\label{r2} M\otimes U\otimes\wedge^{\nnn-1}(A')\rig{} A'\otimes\wedge^{\nnn}(A')
\ene
to obtain
\be\label{r3}
M\otimes U\otimes\wedge^{\nnn-1}(A') \rig{} \wedge^{\nnn}(A').
\ene
Now \eqref{r3} gives a map
\be\label{r4}
 \psi_p':  U\otimes \wedge^{\nnn-1}(A') \rig{} M^{*}\otimes\wedge^{\nnn}(A').
\ene
We claim \eqref{r4} is injective.
(Note that when $\nnn=\mmm$ the source and target space of  \eqref{r4} are dual to each other.)

Consider the transposed map  $S^{\mmm-1}W^{*}\ot \wedge^{\nnn}S^{\mmm+\nnn-2}W\ra S^{\nnn-1}W\ot \wedge^{\nnn-1} S^{\mmm+\nnn-2}W$.
It is defined as follows on decomposable elements (and then extended by linearity):

$$g\ot(f_1\wedge\cdots\wedge f_{\nnn})\mapsto\sum_{i=1}^{\nnn}(-1)^{i-1} g(f_i)\ot f_1\wedge\cdots\hat{f_i}\cdots\wedge f_{\nnn}$$

We show this dual map is surjective. Let $l^{\nnn-1}\ot (l_1^{\mmm+\nnn-2}\wedge\cdots\ww  l_{\nnn-1}^{\mmm+\nnn-2})\in S^{\nnn-1}W\ot \wedge^{\nnn-1} S^{\mmm+\nnn-2}W$
  with $l_i\in W$. Such elements span the target so it will be sufficient to show any such element  is in the image.
 Assume first that
$l$ is distinct from the $l_i$.
Since $\nnn\le\mmm$, there is a polynomial $g\in S^{\mmm-1}W^{*}$ 
which vanishes on $l_1,\ldots , l_{\nnn-1}$ and   is nonzero on $l $.
Then, up to a nonzero scalar, 
$g\ot(l_1^{\mmm+\nnn-2}\wedge\cdots\wedge l_{\nnn-1}^{\mmm+\nnn-2}\wedge l^{\mmm+\nnn-2})$ maps to our element.

Since the image is closed (being a linear space), the condition that
$l$ is distinct from the $l_i$ may be removed  {by taking limits}. 
%Alternatively, say $l_1=l$, then
%take $h\in W$ such that $l,h$ form a basis and $g\in S^{\mmm-1}W^*$ such that $g$ vanishes on $l_1\hd l_{\nnn-1}$ and not on $h$.
%Then $g\ot (l_1^{\mmm+\nnn-2}\wedge\cdots\wedge l_{\nnn-1}^{\mmm+\nnn-2}\wedge (h^{\mmm-1}l^{\nnn-1}))$ maps to the desired vector.

Finally,  $\psi_p'\ot {Id}_L\colon B^*\ot\wedge^{\nnn-1}A'\to C\ot\wedge^\nnn A'$ is  the map induced  {from the restricted matrix multiplication} operator.

To complete the proof of  Theorem \ref{newmainthm}, observe that 
an element of rank one in $A'\ot B\ot C$ induces a map
{$B^*\ot \La{n-1}A'\ra C\ot \La n A'$} of rank ${{\nnn+\mmm-2}\choose{\nnn-1}}$.

{By Lemma \ref{proja} below,  the border rank of $\Mnl$  must be at least the border rank  of  
$T'\in A'\ot B\ot C$, and  by Theorem \ref{resume}  
$$\ur(T') \ge\frac{\dim B^*\otimes \wedge^{\nnn-1}(A') }{{{\nnn+\mmm-2}\choose{\nnn-1}}}=
 \nnn\lll\frac{{{\nnn+\mmm-1}\choose{\nnn-1}}}{{{\nnn+\mmm-2}\choose{\nnn-1}}} = \frac{\nnn\lll\left(\nnn+\mmm-1\right)}{\mmm}.
$$
 }
 This concludes the proof of Theorem \ref{newmainthm}.

\begin{lemma}\label{proja} Let $T\in A\ot B\ot C$, let $A=A'\oplus A''$ and let $\pi\colon A\to A'$ be the 
linear projection,
which induces $\tilde\pi\colon A\ot B\ot C\to A'\ot B\ot C$. {Then  $\bold R(T)\ge \bold R(\tilde\pi(T))$ and  $\ur(T)\ge \ur(\tilde\pi(T))$.}
\end{lemma}
\begin{proof}
If $T=\sum_{i=1}^ra_i\ot b_i\ot c_i$ then
$\tilde\pi(T)=\sum_{i=1}^r\pi(a_i)\ot b_i\ot c_i$.
\end{proof}

\begin{remark} If we let $B'=U$, $C'=M$, then in the proof above we {are} just computing {the rank of} $(T')_A^{\ww p}$ where
  $T'\in A\ot B'\ot C'$ is $Id_U\ot Id_M$.   {  The maximal border rank of a tensor $T$ in
$\BC^{\mmm\nnn}\ot \BC^{\mmm}\ot \BC^{\nnn}$ is $\mmm\nnn$ which occurs anytime the map $T:\BC^{\mmm\nnn*}\ra \BC^{\mmm}\ot \BC^{\nnn}$
is injective, so $T'$}   is a generic tensor
 in $A\ot B'\ot C'$,  and the   calculation of $\trank(\psi_p')$ is  determining the maximal rank of
$(T')_A^{\ww p}$ for a generic element of $\BC^{\mmm\nnn}\ot \BC^{\nnn}\ot \BC^{\mmm}$.
Also note that {the projection $A\to A'$}, viewed as  linear map {$S^{\nnn-1}W^* \ot S^{\nnn-1}W^* \ra    S^{\mmm+\nnn-2}W^*$} is just polynomial  multiplication. 
\end{remark}

%\begin{remark} In terms of matrices, the subspace $A'\subset A$ consists of the so called $\mmm\times\nnn$ catalecticant matrices 
%\red{***add precise reference here***, purpose of this remark??**} of the form

%$$\begin{pmatrix}a_0&a_1&\ldots&a_{\nnn-1}\\
%a_1&&\Ddots&a_\nnn\\
%\vdots&\Ddots&\Ddots&\vdots\\
%a_{\nnn-1}&&&\vdots\\
%\vdots&&&\vdots\\
%a_{\mmm-1}&\ldots&\ldots&a_{\mmm+\nnn-2}
%\end{pmatrix}$$
%for scalars $(a_0,\ldots, a_{\mmm+\nnn-2})$.
%\end{remark}

\section{The kernel through representation theory}\label{MMsect}

We compute the kernel of the map  \eqref{yfmap} as a module and give a formula for its dimension as an alternating sum of products
of binomial coefficients. The purpose of this section  is to show that there are nontrivial equations for tensors of border rank
less than $2\nnn^2$ that matrix multiplication {\it does} satisfy, and to develop a description of the kernel that, we hope, will
be useful for future research.

\subsection{The kernel as a module}   Assume  $\bbb\leq \ccc$, so $\nnn\leq \mmm$.
For a partition $\pi=(\pi_1\hd \pi_N)$, let $\ell(\pi)$ denote the number of parts of $\pi$,
i.e., the largest $k$ such that $\pi_k>0$. Let $\pi'$ denote the conjugate
partition to $\pi$. See \S\ref{glwreps} for the definition of $S_{\pi}U$.

\begin{example}\label{exap4}{C}onsider the case $\mmm=\nnn=3$,
take $p=4$. Let
$$\alpha_1=\begin{matrix}\yng(2,1,1)\end{matrix},\quad \alpha_2=\begin{matrix}\yng(2,2)\end{matrix},\quad \alpha_3=\begin{matrix}\yng(3,1)\end{matrix}$$
Note that $\alpha_1=\alpha_3'$, $\alpha_2=\alpha_2'$.
Then (see \S\ref{decompfor})
 $$
 \wedge^4(M\ot U)=\left(S_{\alpha_3}M\ot S_{\alpha_1}U\right)\oplus
\left(S_{\alpha_2}M\ot S_{\alpha_2}U\right)\oplus
\left(S_{\alpha_1}M\ot S_{\alpha_3}U\right).
$$
Observe that (via the Pieri rule \S\ref{pieris})
\vskip 0.5cm

$\begin{matrix}\yng(2,1,1)\end{matrix}\otimes\begin{matrix}\yng(1)\end{matrix}=\begin{matrix}\yng(3,1,1)\end{matrix}\oplus\begin{matrix}\yng(2,2,1)\end{matrix}$
\vskip 0.5cm

$\begin{matrix}\yng(2,2)\end{matrix}\otimes\begin{matrix}\yng(1)\end{matrix}=\begin{matrix}\yng(3,2)\end{matrix}\oplus\begin{matrix}\yng(2,2,1)\end{matrix}$
\vskip 0.5cm

$\begin{matrix}\yng(3,1)\end{matrix}\otimes\begin{matrix}\yng(1)\end{matrix}=\begin{matrix}\yng(4,1)\end{matrix}\oplus\begin{matrix}\yng(3,2)\end{matrix}\oplus\begin{matrix}\yng(3,1,1)\end{matrix}.$

Among the seven summands on the right-hand side, only $\begin{matrix}\yng(4,1)\end{matrix}$ {does not fit in} the $3\times 3$ square.
{The} kernel of  $M_{3,3,\lll}^{\ww 4}$  in this case is  $L^{*}\ot S_{2,1,1}M\ot S_{4,1}U$,   corresponding to  
 
$\pi=\begin{matrix}\yng(3,1)\end{matrix}\quad\pi+(1)=\begin{matrix}\yng(4,1)\end{matrix}$ 
 which has dimension $\lll \cdot 24\cdot 3=72\lll $.

Let's show that the other two summands in $L^{*}\ot S_{\alpha_1}M\ot S_{\alpha_3}U\ot U$, which are
$L^{*}\ot S_{2,1,1}M\ot S_{3,1,1}U$ and $L^{*}\ot S_{2,1,1}M\ot S_{3,2}U$
are mapped to nonzero elements.

We have (forgetting the identity on $L^*$), the weight vector
\vskip 0.5cm

\newcommand{\muno}{{\small m_1}}
\newcommand{\mdue}{m_2}
\newcommand{\mtre}{m_3}
\newcommand{\mib}{m_i}
\newcommand{\mii}{m^i}
\newcommand{\uuno}{u_1}
\newcommand{\udue}{u_2}
\newcommand{\utre}{u_3}

$\Yboxdim{16pt}\begin{matrix}\young(\muno\muno,\mdue,\mtre)\end{matrix}\otimes\begin{matrix}\young(\uuno\uuno\uuno,\udue,\utre)\end{matrix}$

going to
$\displaystyle\Yboxdim{16pt}\sum_{i=1}^4\begin{matrix}\young(\mii)\end{matrix}\otimes\begin{matrix}\young(\muno\muno\mib,\mdue,\mtre)\end{matrix}\otimes
\begin{matrix}\young(\uuno\uuno\uuno,\udue,\utre)\end{matrix}$, which is nonzero

and the weight vector $\Yboxdim{16pt}\begin{matrix}\young(\muno\muno,\mdue,\mtre)\end{matrix}\otimes\begin{matrix}\young(\uuno\uuno\uuno,\udue\udue)\end{matrix}$
going to

$\displaystyle\Yboxdim{16pt}\sum_{i=1}^4\begin{matrix}\young(\mii)\end{matrix}\otimes\begin{matrix}\young(\muno\muno,\mdue\mib,\mtre)\end{matrix}\otimes
\begin{matrix}\young(\uuno\uuno\uuno,\udue\udue)\end{matrix}$ which is nonzero, too.

Hence the rank of $M_{\langle 3,3,\lll \rangle}^{\ww 4}$ is $3\lll \cdot{9\choose 4}-72\lll = 306\lll$
and $\ur(M_{\langle 3,3,\lll  \rangle}) \ge \lceil\frac{306\lll}{{8\choose 4}}\rceil=\lceil\frac{306\lll}{70}\rceil$ 
which coincides with Lickteig's bound of $14$ when $\lll=3$.  
\end{example}

\begin{lemma}\label{kerislem}  {$\tker  (M_{\langle \mmm,\nnn,\lll\rangle })_{A}^{\ww p}= \oplus_{\pi}  S_{\pi'}M\ot S_{\pi+(1)}U\ot L$ }where
the summation is over partitions $\pi=(\mmm,\nu_1\hd \nu_{\nnn-1})$ where 
$\nu=(\nu_1\hd \nu_{\nnn-1})$ is a partition of $p-\mmm$, $\nu_1\leq \mmm$ and
$\pi+(1)=(\mmm+1,\nu_1\hd \nu_{\nnn-1})$.
\end{lemma}

{\begin{proof}
Write $M_{\mmm,\nnn,\lll}^{\ww p}=\psi_p\ot Id_L$, where 
$\psi_p: \La p (M\ot U)\ot U\ra M^*\ot \La{p+1}(M\ot U)$.
Such modules are contained in the kernel by Schur's lemma, as there is no corresponding module in the target for it to map to
 
 %\red{
% The map
%$U\otimes \Lambda^p(M\otimes U) \ra  M^*\otimes \Lambda^{p+1}(M\otimes U)$
%is the composition of the inclusion into
%$U\otimes (M\otimes U)^{\otimes p}=
%M^{\otimes p}\otimes U^{\otimes p+1}$
%with the injective  map
%$T\mapsto T\otimes Id_M$
%mapping
%$M^{\otimes p}\otimes U^{\otimes p+1}
%\ra
%M^*\otimes M^{\otimes p+1}\otimes U^{\otimes p+1}$,
%with the surjective map projecting down to
%$M^*\otimes \Lambda^{p+1}(M\otimes U)$.

%Write $M_{\mmm,\nnn,\lll}^{\ww p}=\psi_p\ot Id_L$, where 
%$\psi_p: \La p (M\ot U)\ot U\ra M^*\ot \La{p+1}(M\ot U)$.
%It is clear by Schur's lemma that such modules are contained in the kernel, as th%ere is no corresponding module in the target for it to map to.

{We} show that all   {other}  modules in $\La p (M\ot U)\ot U$ are not
in the kernel by computing $\psi_p$ {at  weight} vectors. 
Set  $T'=Id_U\ot Id_M$, so  $\psi_p=(T')_A^{\ww p}$.
Write $T'=(u^i\ot m_\a)\ot m^{\a}\ot u_i$, where  {$1\leq i\leq \nnn$, $1\leq\a\leq \mmm$},
  $(u^i)$ is the dual basis to $(u_i)$ and similarly for $(m^\a)$ and $(m_\a)$, and the summation
convention is used throughout.
Then
$$
T'\ot Id_{\La pA}=(u^i\ot m_\a)\ot m^{\a}\ot u_i\ot [(u_{j_1}\ot m^{\b_1})
\ww\cdots\ww (u_{j_p}\ot m^{\b_p})]
\ot 
[(u^{j_1}\ot m_{\b_1})
\ww\cdots\ww (u^{j_p}\ot m_{\b_p})]
$$
and
$$
(T')_A^{\ww p}=
[(u_{j_1}\ot m^{\b_1})
\ww\cdots\ww (u_{j_p}\ot m^{\b_p})]
\ot u_i\bigotimes 
  m^{\a}\ot  
[(u^{j_1}\ot m_{\b_1})
\ww\cdots\ww (u^{j_p}\ot m_{\b_p})\ww (u^i\ot m_\a)]
$$
Note that all the summands of the decomposition (see (\ref{cauchy}))
 $\displaystyle\La p (M\ot U)=\oplus_{|\alpha|=p}(S_{\alpha}M\otimes S_{\alpha'}U)$ are multiplicity free,
it follows that also the summands of $\displaystyle\La p (M\ot U)\ot U$ are multiplicity free.
So we can compute $(T')_A^{\ww p}$ considering one weight vector for any irreducible summand.
Under the mapping of a weight vector, nothing happens to the $U$ component. The $M$ component gets
tensored with the identity map and then projected onto the component that gives the conjugate diagram to
the one of $U$.
So it just remains to see this projection is nonzero. But the projection is just as in Example \ref{exap4}, we add a box
in the appropriate place and sum over basis vectors. Since the diagram will be one that produces a nonzero
module for $M$, at least one basis vector can be placed in the new box to yield a nonzero Young tableaux.
But we are summing over all basis vectors.
\end{proof}
 
\subsection{{Dimension of the kernel}} \label{rbndpfsect}
 We compute the dimension  of  $\tker  (M_{\langle \mmm,\nnn,\lll\rangle })_{A}^{\ww p}=
\tker \psi_p\ot Id_L$  via an exact sequence. 
We continue the notations of above.
Consider the map
\begin{align}
\psi_{p,2}: \La{p-\mmm}(M\ot U)\ot \La \mmm M\ot S^{\mmm+1}U&\ra \La p (M\ot U)\ot U\\
T\ot m_1\ww\cdots\ww m_{\mmm}\ot u^{\mmm+1}&\mapsto T\ww (m_1\ot u)\ww\cdots{\ww} (m_{\mmm}\ot u)\ot u.
\end{align}

\begin{lemma}\label{kerpsiplem}
$\tim \psi_{p,2}=\tker \psi_p$.
\end{lemma}
\begin{proof}
Observe that
$$
\psi_p: \bigoplus_{\stack{|\pi|=p,\ell(\pi)\leq \nnn}{\pi_1\leq \mmm}}
S_{\pi}U\ot U\ot S_{\pi'}M
\ra
  \bigoplus_{\stack{|\mu|=p+1,\ell(\mu)\leq \mmm }{\mu_1\leq \nnn}}
S_{\mu}M\ot M^*\ot S_{\mu'}U
$$
  {is}  a $GL(U)\times GL(M)$-module map.
Now the source of $\psi_{p,2}$ is
$$
\bigoplus_{\stack{|\nu|=p-\mmm, \nu_1\leq \mmm}{\ell(\nu)\leq \nnn}}
S_{\nu'}M\ot S_{\nu}U\ot S^{\mmm+1}U\ot \La \mmm M
$$
and a given   module in the source with {$\nu_{\nnn}=0$} maps to
$S_{\pi+(1)}U\ot S_{\pi'}M\subset S_{\pi}U\ot U\ot S_{\pi'}M$ where $\pi=({\mmm},\nu_1\hd\nu_{\nnn-1})$, the proof is
similar to the proof of Lemma \ref{kerislem}. Its other components map to zero.
\end{proof}

 The kernel of  $\psi_{p,2}$ is  
the image of
\begin{align}
\psi_{p,3}: \La{p-{\mmm}-1}(M\ot U)\ot \La {\mmm} M\ot M\ot S^{{\mmm}+2}U&\ra \La{p-{\mmm}}(M\ot U)\ot \La {\mmm} M\ot S^{{\mmm}+1}U\\
\nonumber T\ot m_1\ww\cdots\ww m_{\mmm}\ot m\ot u^{{\mmm}+2}&\mapsto T\ww (m \ot u)\ot m_1\ww\cdots\ww m_{\mmm} \ot u^{{\mmm}+1}
\end{align}
and $\psi_{p,3}$ has kernel the image of
\begin{align}
\psi_{p,4}: \La{p-{\mmm}-2}(M\ot U)\ot \La {\mmm} M\ot S^2M\ot S^{{\mmm}+3}U&\ra \La{p-{\mmm}-1}(M\ot U)\ot \La {\mmm} M\ot M\ot S^{{\mmm}+2}U\\
\nonumber T\ot m_1\ww\cdots\ww m_{\mmm}\ot m^2\ot u^{{\mmm}+3}&\mapsto T\ww (m\ot u)\ot m_1\ww\cdots\ww m_{\mmm} \ot m\ot u^{{\mmm}+2}
\end{align}
One defines analogous maps $\psi_{p,k}$.
By taking the Euler characteristic we obtain:
\begin{lemma}  \label{kerpsipsumlem}
$$
\tdim\tker\psi_p=
\sum_{j=0}^{p-\mmm } (-1)^j\binom{\mmm\nnn}{p-\mmm-j}
\binom{\mmm+j-1}j\binom{\mmm+\nnn+j}{\mmm+j+1}.
$$
 \end{lemma}

 In summary: 
 
\begin{theorem}\label{oldmainthm} Set $p\leq \lceil \frac {\mmm\nnn}2\rceil-1$ and assume  $\nnn\le\mmm$. Then
$$
 \tdim(\tker(\Mnl)_{A}^{\ww p})=
 \lll  \sum_{j=0}^{p-\mmm } (-1)^j\binom{\mmm\nnn}{p-\mmm-j}
\binom{\mmm+j-1}j\binom{\mmm+\nnn+j}{\mmm+j+1} .
$$ 
%In the case $\mmm=\nnn$,  
%set $q=2\cdot\lceil\frac{\nnn(\nnn+1)}{4}\rceil -1$,
%and set $p=\lceil\frac{q}{2}\rceil-1$. Then
%$$
%\tdim(\tker((\Mn)_{A}^{\ww p}))=\nnn  \sum_{j=0}^{p-\nnn } (-1)^j\binom{q}{p-\nnn-j}
%\binom{\nnn+j-1}j\binom{2\nnn+j}{\nnn+j+1} .
%$$
\end{theorem}
 
%  The second part is proved by making an identification $M\simeq U$ and
%restricting to $A'=S^2U\subset U\ot U$ or $A'=S^2_0U$, the traceless symmetric matrices.

In the case $\mmm=\nnn$ one can get a smaller kernel by identifying $V^*\simeq U$ and restricting to $A'=S^2U\subset U\ot U$,
although this does not appear to give a better lower bound than Theorem \ref{mainthm2}. A different restriction that allows for
a small kernel could conceivably give a better bound.

\section{Review of Lickteig's bound}\label{licksect}
For comparison, we outline the proof of Lickteig's bound.  (Expositions of Strassen's bound are given in several places, e.g. 
\cite[Chap. 3]{Ltensor} {and \cite[\S 19.3]{bucs:96}.})  It follows in three steps.
The first combines two  standard facts from algebraic geometry: for varieties
$X,Y\subset \BP V$, let 
$J(X,Y)\subset \BP V$ denote the join of $X$ and $Y$. Then
$\s_{r+s}(X)=J(\s_r(X),\s_s(X))$. If $X=Seg(\BP A\times \BP B\times \BP C)$
is a Segre variety, then $\s_s(Seg(\BP A\times \BP B\times \BP C))\subseteq Sub_s(A\ot B\ot C)$,
where 
\begin{align*}Sub_s(A\ot B\ot C):=\{& T\in A\ot B\ot C \mid 
\\
&\exists A'\subset A,\ B'\subset B, C'\subset C,
\ \tdim A'=\tdim B'=\tdim C'=s, \ T\in A'\ot B'\ot C'\}.
\end{align*}
 See, e.g.,\cite[{\S 7.1.1}]{Ltensor} for details.
Next Lickteig observes that if $T\in \s_{r+s}(Seg(\BP A\times \BP B\times \BP C))$, then
there exist $A',B',C'$ each of dimension $s$ such that, thinking of
$T:A^*\ot B^*\ra C$, 
\be\label{licke}
\tdim (T( ( A')\upperp \ot B^*+ A^*\ot (B')\upperp))\leq r.
\ene 
This follows because the condition is a  closed condition and it holds for
points on the open subset of points in the span of $r+s$ points on $Seg(\BP A\times \BP B\times \BP C)$.

Finally, for matrix multiplication, with $A=M\ot N^*$ etc., he defines $M'\subset M$, ${N^*}'\subset N^*$ to
be the smallest spaces such that $A'\subseteq M'\ot {N^*}'$ and similarly for the other spaces.
Then one applies \eqref{licke} combined with the observation that
$M|_{(A')\upperp\ot B^*}\subseteq M'\ot L^*$ etc., and keeps track of the various bounds to conclude.

\section{Appendix: facts from representation theory}
\label{repapp}

\subsection{Representations of $GL(V)$}\label{glwreps}
The irreducible  representations of $GL(V)$ are indexed by sequences $\pi=(p_1\hd p_l)$ of non-increasing integers with $l\leq \tdim V$. 
Those that occur in $V^{\ot d}$ are partitions of $d$, and we write $|\pi|=d$ and   $S_{\pi}V$ for the module.
$V^{\ot d}$ is also an $\FS_d$-module, and  the groups $GL(V)$ and $\FS_d$ are the commutants of each other in $V^{\ot d}$ which implies
the famous Schur-Weyl duality that $V^{\ot d}=\oplus_{|\pi|=d,\ell(\pi)\leq \bv} S_{\pi}V\ot[\pi] $ as a
$(GL(V)\times \FS_d)$-module, where $[\pi]$ is the irreducible
$\FS_d$-module associated to $\pi$. Repeated numbers in partitions are sometimes expressed as exponents when there is no danger of
confusion, e.g. $(3,3,1,1,1,1)=(3^2,1^4)$.
For example, $S_{(d)}V=S^dV$ and $S_{(1^d)}V=\La d V$. The modules  $S_{s^{\bv}}V=(\La {\bv}V)^{\ot s}$ 
  are trivial as $SL(V)$-modules .
The module $S_{(22)}V$ is the home of the Riemann curvature tensor in Riemannian geometry.
 See any of \cite[Chap. 6]{Ltensor}, \cite[Chap 6]{FH} or \cite[Chap. 9]{MR2265844} for more details on the
representations of $GL(V)$  and what follows.

\subsection{Useful decomposition formulas}\label{decompfor}
To decompose $S^2(A\ot B)$ as a $GL(A)\times GL(B)$-module, note that
given $P\in S^2A$ and $Q\in S^2B$, the   product  of $P$ and $Q$ defined by
$P\ot Q(\a\ot \b,\a'\ot \b'):=P(\a,\a')Q(\b,\b')$ will be in $S^2(A\ot B)$. Similarly,
if $P\in \La 2 A$ and $Q\in \La 2 B$,   $P\ot Q$  
will also be symmetric as $P(\a',\a)Q(\b',\b)=[-P(\a,\a')][-Q(\b,\b')]=P(\a,\a')Q(\b,\b')$.
Since the dimensions of these spaces add to the dimension of $S^2(A\ot B)$ we 
conclude {$$S^2(A\ot B)=(S^2A\ot S^2B)\op (\La 2 A\ot \La 2 B).$$}
By an analogous argument, we have the decomposition {$$\La 2 (A\ot B)= (S^2A\ot \La 2B)\op (\La 2 A\ot S^2 B).$$}
More generally (see, e.g. \cite[\S 6.5.2]{Ltensor}) we have
\begin{align}\label{cauchy}
\La p(A\ot B)&=\oplus_{|\pi|=p} S_{\pi}A\ot S_{\pi'}B\\
S^p(A\ot B)&=\oplus_{|\pi|=p} S_{\pi}A\ot S_{\pi}B
\end{align}
where $\pi'$ denotes the conjugate partition to $\pi$, that is, if we represent $\pi=(p_1\hd p_n)$ by a Young diagram,
with $p_j$ boxes in the $j$-th row, the diagram of $\pi'$ is obtained by reflecting the diagram of $\pi$ along the $NW$ to $SE$
axis.

\subsection{The Pieri rule}\label{pieris}
The decomposition of $S_{\pi}V\ot V$ is multiplicity free, consisting of a copy of each $S_{\mu}V$
such that the Young diagram of $\mu$ is obtained from the Young diagram of $\pi$ by adding a box.
(Boxes must be added in such a way that one still has a partition and the number of rows is at most
the dimension of $V$.)

For example:

$\begin{matrix}\yng(2,1,1)\end{matrix}\otimes\begin{matrix}\yng(1)\end{matrix}=\begin{matrix}\yng(3,1,1)\end{matrix}\oplus\begin{matrix}\yng(2,2,1)\end{matrix}$
\vskip 0.5cm

More generally, 
  {\it Pieri formula } states that $S_{\pi}V\ot S^dV$ decomposes multiplicity free into the sum of all
  $S_{\mu}V$ that can be obtained by adding $d$ boxes to the Young diagram of $\pi$ in such a way that
  no two boxes are added to the same column, and 
  $S_{\pi}V\ot \La dV$ decomposes multiplicity free into the sum of all
  $S_{\mu}V$ that can be obtained by adding $d$ boxes to the Young diagram of $\pi$ in such a way that
  no two boxes are added to the same row. See any of the standard references given above for details.

\bibliographystyle{amsplain}
 
\bibliography{LmatrixLO}

\end{document}